\documentclass[review]{elsarticle}

\usepackage{lineno,hyperref}
\modulolinenumbers[5]

\journal{}

\usepackage[utf8]{inputenc}
\usepackage{xcolor,colortbl}
\usepackage{xspace}
\usepackage{url}
\usepackage{graphicx}
\usepackage[ruled, linesnumbered, vlined, lined, boxed]{algorithm2e}
\usepackage{booktabs}
\usepackage{amsmath}
\usepackage{arydshln}
\usepackage{wasysym}

\newcommand{\SA}{\ensuremath{\mathsf{SA}}\xspace}
\newcommand{\LA}{\ensuremath{\mathsf{\lambda}}\xspace}
\newcommand{\LABP}{\ensuremath{\mathsf{\lambda}_{BP}}\xspace}
\newcommand{\LF}{\ensuremath{\mathsf{LF}}\xspace}
\newcommand{\ISA}{\ensuremath{\mathsf{ISA}}\xspace}
\newcommand{\NSV}{\ensuremath{\mathsf{NSV}}\xspace}
\newcommand{\BWT}{\ensuremath{\mathsf{BWT}}\xspace}

\renewcommand{\L}{\ensuremath{L}\xspace}

\newcommand{\etal}{{\it et al.}\xspace}

\newcommand{\our}{{\textsf{BWT-Lyndon}}\xspace}
\newcommand{\nsv}{{\textsf{NSV-Lyndon}}\xspace}
\newcommand{\baier}{\textsf{Baier-Lyndon}\xspace}
\newcommand{\maxlyn}{\textsf{MaxLyn}\xspace}

\newcommand{\p}{\phantom{0}}

\newcommand{\und}{\underline}

\newcommand{\rank}{\mathsf{rank}}
\newcommand{\sel}{\mathsf{select}}
\newcommand{\selo}{\mathsf{selectopen}}
\newcommand{\findc}{\mathsf{selectclose}}
\newcommand{\Oh}{O}
\newcommand{\opa}{ \mbox{\tt\small (}  }
\newcommand{\cpa}{\mbox{\tt\small )} }
\long\def\ignore#1{\vskip 0pt}

\newcommand{\opab}{ \mbox{\tt\small "("}  }
\newcommand{\cpab}{\mbox{\tt\small ")"} }

\newcommand{\qedwhite}{\hfill \ensuremath{\Box}}

\newtheorem{theorem}{Theorem}
\newtheorem{lemma}{Lemma}
\newdefinition{rmk}{Remark}
\newproof{proof}{Proof}

\begin{document}

\begin{frontmatter}

\title{Lyndon Array Construction during Burrows-Wheeler~Inversion}



\author[usp]{Felipe A. Louza\corref{mycorrespondingauthors}}
\cortext[mycorrespondingauthors]{Corresponding author}
\ead{louza@usp.br}
\author[macmaster,murdoch]{W. F. Smyth}
\ead{smyth@mcmaster.ca}
\author[unipo,cnr]{Giovanni Manzini}
\ead{giovanni.manzini@uniupo.it}
\author[unicamp]{Guilherme P. Telles}
\ead{gpt@ic.unicamp.br}

\address[usp]{Department of Computing and Mathematics, University of São Paulo, Brazil}
\address[macmaster]{Department of Computing and Software, McMaster University, Canada}
\address[murdoch]{School of Engineering \& Information Technology, Murdoch University, Australia}
\address[unipo]{Computer Science Institute, University of Eastern Piedmont, Italy}
\address[cnr]{Institute of Informatics and Telematics, CNR, Pisa, Italy}
\address[unicamp]{Institute of Computing, University of Campinas, Brazil}

\begin{abstract}
In this paper we present an algorithm to compute the Lyndon array of a
string~$T$ of length~$n$ as a byproduct of the inversion of the
Burrows-Wheeler transform of $T$. Our algorithm runs in linear time using
only a stack in addition to the data structures used for Burrows-Wheeler
inversion. We compare our algorithm with two other linear-time algorithms for
Lyndon array construction and show that computing the Burrows-Wheeler
transform and then constructing the Lyndon array is competitive compared to
the known approaches. We also propose a new balanced parenthesis
representation for the Lyndon array that uses $2n+o(n)$ bits of space and
supports constant time access. This representation can be built in linear
time using $O(n)$ words of space, or in $O(n\log n/\log\log n)$ time using
asymptotically the same space as $T$.
\end{abstract}

\begin{keyword}
Lyndon array, Burrows-Wheeler inversion, linear time, compressed
representation, balanced parentheses.
\end{keyword}

\end{frontmatter}


\section{Introduction}
\label{s:intro}

Lyndon words were introduced to find bases of the free Lie algebra
\cite{CFL58}, and have been extensively applied in algebra and combinatorics.
The term ``Lyndon array'' was apparently introduced in \cite{Franek2016},
essentially equivalent to the ``Lyndon tree'' of Hohlweg \& Reutenauer
\cite{Hohlweg2003}. Interest in Lyndon arrays has been sparked by the
surprising characterization of runs through Lyndon words by Bannai
\etal~\cite{Bannai2014}, who were thus able to resolve the long-standing
conjecture that the number of runs (maximal periodicities) in any string of
length $n$ is less than $n$.

The Burrows-Wheeler transform (BWT)~\cite{Burrows1994} plays a fundamental
role in data compression and in text
indexing~\cite{Ohlebusch2013,Makinen2015,Navarro2016}.  Embedded into a
wavelet tree, the BWT is a self-index with a remarkable time/space
tradeoff~\cite{Ferragina2000,Grossi2003}.

In this article we introduce a linear time algorithm to construct the Lyndon
array of a string $T$ of length $n$, from an ordered alphabet of
size~$\sigma$, as a byproduct of Burrows-Wheeler inversion, thus establishing
an apparently unremarked connection between BWT and Lyndon array
construction. We compare our algorithm to others in the literature that also
compute the Lyndon array in worst-case linear time. We find that the new
algorithm performs well in practice with a small memory footprint.

Inspired by the inner working of our new algorithm, we propose a
representation of the Lyndon array consisting of a balanced parenthesis
string of length $2n$. Such representation leads to a data structure of size
$2n + o(n)$ bits, supporting the computation of each entry of the Lyndon
array in constant time. We also show that such representation is
theoretically appealing since it can be computed from $T$ in $O(n)$ time
using $O(n)$ words of space, or in $O(n\log n/\log\log n)$ time using $O(n
\log\sigma)$ bits of space.

This article is organized as follows. Section~\ref{s:prelim} introduces
concepts, notation and related work. Section~\ref{s:algorithm} presents our
algorithm and Section~\ref{s:experiments} shows experimental results.
Section~\ref{s:bp} describes our balanced parenthesis representation of the
Lyndon array and two construction algorithms with different time/space
tradeoffs. Section~\ref{s:conclusion} summarizes our conclusions.

\section{Concepts, notation and related work}
\label{s:prelim}

Let $T$ be a string of length $|T|=n$ over an ordered alphabet $\Sigma$ of
size $\sigma$.  The $i$-th symbol of $T$ is denoted by $T[i]$ and the
substring $T[i] T[i+1]\cdots T[j]$ is denoted by $T[i,j]$, for $1 \le i \le j
\le n$. We assume that $T$ always ends with a special symbol $T[n] = \$$,
that doesn't appear elsewhere in $T$ and precedes every symbol in $\Sigma$. A
prefix of $T$ is a substring of the form $T[1,i]$ and a suffix is a substring
of the form $T[i,n]$, which will be denoted by $T_i$.  We use the symbol
$\prec$ for the lexicographic order relation between strings.

The suffix array (\SA)~\cite{Manber1993,Gonnet1992} of a string $T[1,n]$ is
an array of integers in the range $[1, n]$ that gives the lexicographic order
of all suffixes of $T$, such that $T[\SA[1],n] \prec T[\SA[2],n] \prec \cdots
\prec T[\SA[n],n]$.  We denote the inverse of \SA as \ISA,
$\ISA[\SA[i]]=i$. The suffix array can be constructed in linear time using
$O(\sigma)$ additional space~\cite{Nong2013}.

The next smaller value array ($\NSV_A$) defined for an array of integers $A[1,n]$
stores in $A[i]$ the position of the next value in $A[i+1,n]$ that
is smaller than $A[i]$. If there is no value in $A[i+1,n]$ smaller than
$A[i]$ then $\NSV_A[i] = n+1$. Formally, $\NSV_A[i] = \min(\{n+1\}\cup\{j | i <
j \leq n \mbox{ and } A[j] < A[i]\})$. \NSV may be constructed in
linear time using additional memory for an auxiliary stack~\cite{Goto2013}.

\paragraph{Lyndon array}
A string $T$ of length $n>0$ is called a Lyndon
word if it is lexicographically strictly smaller than its
rotations~\cite{CFL58}. Alternatively, if $T$ is a Lyndon word and $T=uv$ is any
factorization of $T$ into non-empty strings, then $u\prec v$.  The Lyndon
array of a string $T$, denoted $\LA_T$ or simply $\LA$ when $T$ is
understood, has length $|T|=n$ and stores at each position $i$ the length of
the longest Lyndon word starting at $T[i]$.

\ignore{The Lyndon array can also be defined to store the end position of the
longest Lyndon word starting at $T[i]$, denoted by $\LA_{end}$.}

Following \cite{Hohlweg2003}, Franek \etal~\cite{Franek2016} have recently
shown that the Lyndon array can be easily computed in linear time by applying
the \NSV computation to the inverse suffix array (\ISA), such that $\LA[i] =
\NSV_{\ISA}[i]-i$, for $1 \leq i \leq n$. Also, in a recent talk surveying
Lyndon array construction, Franek and Smyth~\cite{Smyth2017} quote
an unpublished observation by Cristoph Diegelmann~\cite{Diegelmann2016} that,
in its first phase, the linear-time suffix array construction algorithm by
Baier~\cite{Baier2016} computes a permuted version of the Lyndon array. This
permuted version, called $\LA_{\SA}$, stores in $\LA_{\SA}[i]$ the length of
the longest Lyndon word starting at position $\SA[i]$ of $T$. Thus, including
the BWT-based algorithm proposed here, there are apparently three algorithms
that compute the Lyndon array in worst-case $O(n)$ time. In addition, in
\cite[Lemma 23]{Bannai2014} a linear-time algorithm is suggested that uses
lca/rmq techniques to compute the Lyndon tree. The same paper also gives an
algorithm for Lyndon tree calculation described as being ``in essence'' the
same as \NSV.

\paragraph{Burrows-Wheeler transform}
The Burrows-Wheeler transform (BWT)~\cite{Burrows1994,Adjeroh2008} is a
reversible transformation that produces a permutation $L$ of the original
string $T$ such that equal symbols of $T$ tend to be clustered in $L$.
The BWT can be obtained by adding each circular shift of $T$ as a row of a
conceptual matrix $M'$, lexicographically sorting the rows of $M'$ producing
$M$, and concatenating the symbols in the last column of $M$ to form $L$.
Alternatively, the BWT can be obtained from the suffix array through the
application of the relation $L[i] = T[SA[i]-1]$ if $SA[i]\neq 1$ or $L[i] =
\$$ otherwise.

Burrows-Wheeler inversion, the processing of $L$ to obtain $T$, is based on the
LF-mapping (last-to-first mapping).  Let $c_F$ and $c_L$ be the first and the
last columns of the conceptual matrix $M$ mentioned above.
We have
$LF:\{1,\dots,n\}\rightarrow \{1,\dots,n\}$ such that if $c_L[i]=\alpha$ is the
$k^{th}$ occurrence of a symbol $\alpha$ in $c_L$, then $LF(i)=j$ corresponds to
the position $c_F[j]$ of the $k^{th}$ occurrence of $\alpha$ in $c_F$.

\begin{figure}[t]
\centering
\resizebox{1\textwidth}{!}{
\begin{tabular}{c|c|c|c|c|c|c|c|c|c|l}
  \multicolumn{1}{c}{} & \multicolumn{1}{c}{} & \multicolumn{1}{c}{sorted} & \multicolumn{7}{c}{} & \multicolumn{1}{c}{sorted} \\
  \multicolumn{1}{c}{} & \multicolumn{1}{c}{circular shifts} & \multicolumn{1}{c}{circular shifts} & \multicolumn{7}{c}{} & \multicolumn{1}{c}{suffixes}\\
  \hline
  $i$ &                  & {$F$~~~~~~~$L$} &  \SA & \ISA &$\NSV_{\ISA}$ & \LF & $\LA$  & $\LA_\SA$ & $\L$ & $T[\SA[i],n]$    \\
  \hline
 1  & \texttt{\und{b}anana\$}  & \texttt{\$banana}  & 7 & 5 & 2 & 2 & 1 & 1 & \texttt{a}  & \texttt{~\und{\$}}        \\
 2  & \texttt{\und{an}ana\$b}  & \texttt{a\$banan}  & 6 & 4 & 4 & 6 & 2 & 1 & \texttt{n}  & \texttt{~\und{a}\$}       \\
 3  & \texttt{\und{n}ana\$ba}  & \texttt{ana\$ban}  & 4 & 7 & 4 & 7 & 1 & 2 & \texttt{n}  & \texttt{~\und{an}a\$}     \\
 4  & \texttt{\und{an}a\$ban}  & \texttt{anana\$b}  & 2 & 3 & 6 & 5 & 2 & 2 & \texttt{b}  & \texttt{~\und{an}ana\$}   \\
 5  & \texttt{\und{n}a\$bana}  & \texttt{banana\$}  & 1 & 6 & 6 & 1 & 1 & 1 & \texttt{\$} & \texttt{~\und{b}anana\$}  \\
 6  & \texttt{\und{a}\$banan}  & \texttt{na\$bana}  & 5 & 2 & 7 & 3 & 1 & 1 & \texttt{a}  & \texttt{~\und{n}a\$}      \\
 7  & \texttt{\und{\$}banana}  & \texttt{nana\$ba}  & 3 & 1 & 8 & 4 & 1 & 1 & \texttt{a}  & \texttt{~\und{n}ana\$}    \\
  \hline
\end{tabular}
}
\caption{Circular shifts, sorted circular shifts, \SA, \ISA, $\NSV_{\ISA}$, \LF, $\LA_\SA$, \LA,
  $L$ and the sorted suffixes of $T=$ \texttt{banana\$}.}
\label{fig:bwt}
\end{figure}

The $LF$-mapping can be pre-computed in an array \LF of integers in the
range $[1,n]$.  Given an array of integers $C$ of length $\sigma$ that stores
in $C[\alpha]$ the number of symbols in $T$ strictly smaller than $\alpha$,
\LF can be computed in linear time using $O(\sigma \log n)$ bits of
additional space~\cite[Alg.  7.2]{Ohlebusch2013}.  Alternatively, $LF(i)$ can
be computed on-the-fly in $O(\log \sigma)$ time querying a wavelet
tree~\cite{Grossi2003} constructed for $c_L$. Given the \BWT \L and the \LF
array, the Burrows-Wheeler inversion can be performed in linear time~\cite[Alg.
7.3]{Ohlebusch2013}.

Figure~\ref{fig:bwt} shows the circular shifts, the sorted circular shifts,
the arrays \SA, \ISA, $\NSV_{\ISA}$, \LF, \LA, $\LA_\SA$, the \BWT \L and the
sorted suffixes of $T=$ \texttt{banana\$}.  The longest Lyndon words starting
at each position $i$ and $\SA[i]$ are underlined in the first and last
columns of Figure~\ref{fig:bwt}.

\section{From the BWT to the Lyndon array}\label{s:algorithm}

Our starting point is the following characterization of the Lyndon array.

\begin{lemma}\label{l:suffixes}
Let $j$ be the smallest position in $T$ after position $i<n$ such that suffix
$T[j,n]$ is lexicographically smaller than suffix $T[i,n]$; that is, $j =
\min\{k | i < k \le n \mbox{ and } T[k,n] \prec T[i,n]\}$. Then the length of
the longest Lyndon word starting at position $i$ is $\LA[i] = j-i$.  If $i=n$
then $\LA[i]=1$.
\end{lemma}

\begin{proof} For $i<n$ let $j$ be defined as above and $w=T[i,j-1]$. If
$w=uv$ then $\exists h$, $i<h<j$, such that $u=T[i,h-1]$ and $v=T[h,j-1]$.
Since $h<j$ it follows that $T[h,n]\succ T[i,n]$, hence $v \succ u$ and $T[i,j-1]$ is a
Lyndon word. In addition, $T[j,j] \prec T[j,n] \prec T[i,n]$, hence $T[j]
\leq T[i]$ and $T[i,j]$ is not a Lyndon word. \qedwhite
\end{proof}

The above lemma is at the basis of the algorithm by Franek
\etal~\cite{Franek2016} computing $\LA[i]$ as $\NSV_{\ISA}[i]-i$. Since
$\ISA[i]$ is the lexicographic rank of $T[i,n]$, $j=\NSV_{\ISA}[i]$ is
precisely the value used in the lemma. In this section, we use the
relationship between $LF$-mapping and $\ISA$ to design alternative
algorithms for Lyndon array construction. Since $\ISA[n]=1$, and $LF(\ISA[i])
= \ISA[i-1]$ it follows that $\ISA[i] = LF^{n-i}(n)$ where $LF^j$ denotes the $LF$ map
iterated $j$ times.

Given the \BWT \L and the $LF$ mapping our algorithm computes $T$ and the
Lyndon array $\lambda$ from right to left. Briefly, our algorithm finds,
during the Burrows-Wheeler inversion, for each position $i=n, n-1, \dots, 1$,
the first suffix $T[j,n]$ that is smaller than $T[i,n]$ and using
Lemma~\ref{l:suffixes} it computes $\LA[i] = j-i$.

Starting with $i=n$ and an index ${\it pos}=1$ in the BWT, the algorithm
decodes the BWT according to $LF(pos)$, keeping the visited positions whose
indexes are smaller than {\it pos} in a stack. The visited positions indicate
the suffix ordering: a suffix visited at position $i$ is lexicographically
smaller than all suffixes visited at positions $j>i$. The stack stores pairs
of integers $\big<pos,step\big>$ corresponding to each visited position $pos$
in iteration $step$.  The stack is initialized by pushing $\big<-1,0\big>$.
The complete pseudo-code appears in Algorithm~\ref{a:lyndon-bwt}. Note that
lines 1, 2, 7, 8, 15 and 16 are exactly the lines from the Burrows-Wheeler
inversion presented in~\cite[Alg. 7.3]{Ohlebusch2013}.

An element $\big<pos,step\big>$ in the stack represents the suffix
$T[n-step+1,n]$ visited in iteration $step$.  At iteration $step$ the
algorithm pops suffixes that are lexicographically larger than the current
suffix $T[n-step+1,n]$.  Consequently, at the end of the while loop, the top
element represents the next suffix (in text order) that is smaller than
$T[n-step+1,n]$ and $\LA[step]$ is computed at line~\ref{a:lambda}.

\begin{algorithm}[t]
\SetNlSty{textbf}{}{} \SetAlgoLined \SetCommentSty{mycommfont}
\KwData{$L[1,n]$ and $\LF[1,n]$} \KwResult{$T[1,n]$ and $\lambda[1,n]$}

$T[n] \leftarrow \$$

$pos \leftarrow 0$

$\lambda[n] \leftarrow 1$

$Stack \leftarrow \emptyset$

$Stack.push(\big<-1, 0\big>)$

$step \leftarrow 1$

\For{$i \leftarrow n-1\mbox{ \bf downto }1$}{

        $T[i] \leftarrow \L[pos]$

        \While{$Stack.top().pos > pos$}{
                $Stack.pop()$
        }

        $\lambda[i] \leftarrow step - Stack.top().step$\label{a:lambda}

        $Stack.push(\big<pos, step\big>)$

        $step \leftarrow step+1$

        $pos \leftarrow LF[pos]$
} \caption{Lyndon array construction during Burrows-Wheeler inversion}
\label{a:lyndon-bwt}
\end{algorithm}

\paragraph{Example}
Figure~\ref{f:example} shows a running example of our algorithm to compute
the Lyndon array for string $T=\texttt{banana\$}$ during its Burrows-Wheeler
inversion.
Before {\it step} is set to $1$ (lines 1--6) $\$$ is decoded at position $n$ and the
stack is initialized with the end-of-stack marker $\big<-1, 0\big>$. The
first loop iteration (lines 7--15) decodes $\texttt{a}$ and finds out that
the stack is empty. Then $\LA[6]=1$, the pair $\big<1, 1\big>$ is pushed on
the stack and $pos=LF[1]=2$.

At the second iteration $\texttt{n}$ is decoded and the algorithm checks if
the suffix at the top of the stack ($\texttt{a\$}$) is larger then the
current suffix ($\texttt{na\$}$). The algorithm does not pop the stack
because there is no suffix lexicographically larger than the current one.
Then $\LA[5] = step-Stack.top().step = 2-1 = 1$. The pair $\big<6, 2\big>$ is
pushed on the stack.

At the third iteration $\texttt{a}$ is decoded. The top element, representing
suffix $\texttt{na\$}$, is popped since it is larger then the current suffix
$\texttt{ana\$}$. Then $\LA[4] = step-Stack.top().step = 3-1 = 2$ and the
pair $\big<3, 3\big>$ is pushed. The next iterations proceed in a similar
fashion.

\begin{figure}[t]
\centering
\includegraphics[width=1\textwidth]{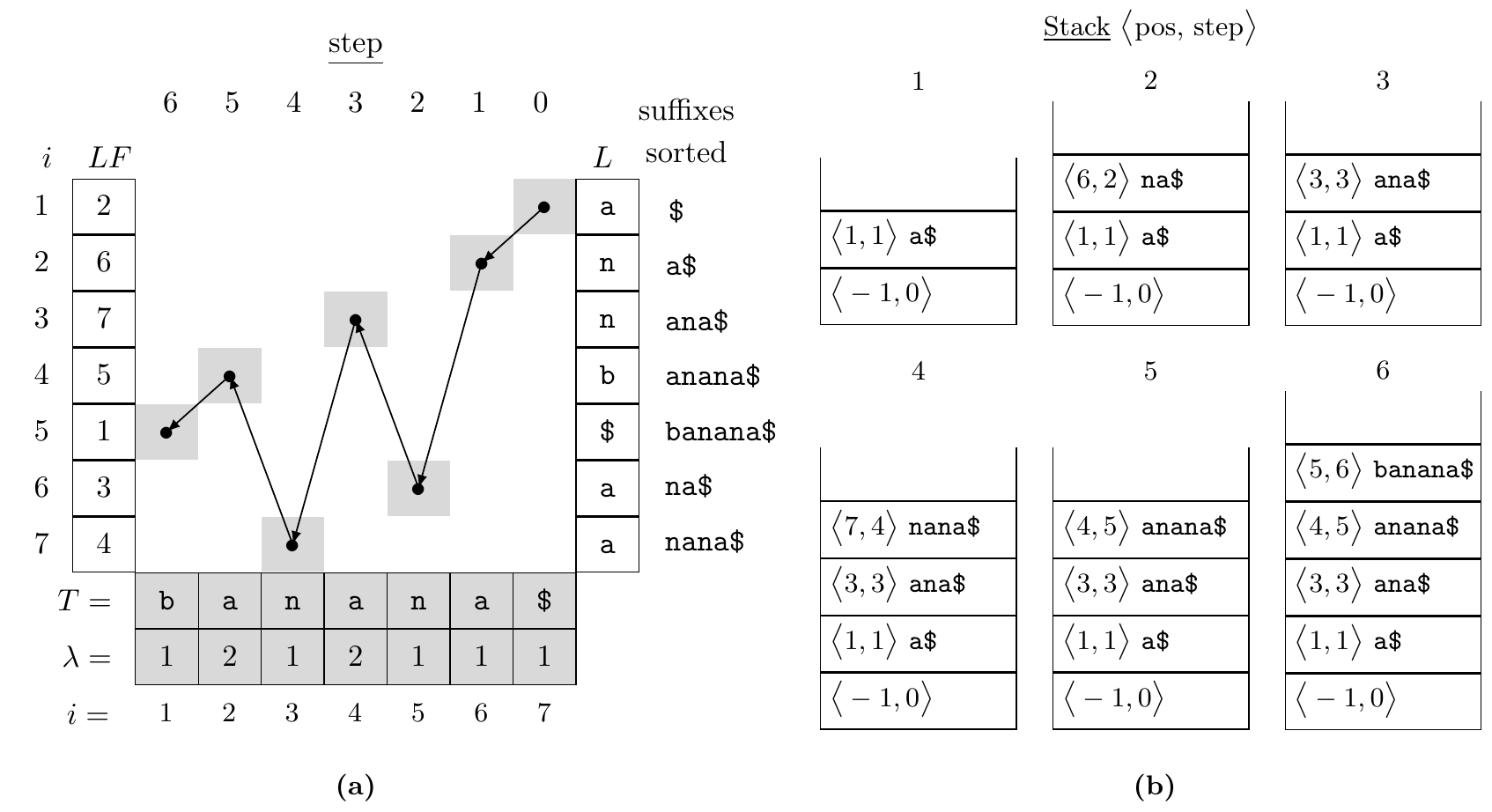}
\caption{Example of our algorithm in the string $T=\texttt{banana\$}$.
Part (a) shows the algorithms steps from right to left.
The arrows illustrate the order in which suffixes are visited
by the algorithm, following the LF-mapping.
Part (b) shows the Stack and the corresponding suffixes at
the end of each step of the algorithm.
}
\label{f:example}
\end{figure}

\begin{lemma}\label{l:lyn_bwt}
Algorithm~\ref{a:lyndon-bwt} computes the text $T[1,n]$ and its Lyndon
array~$\LA[1,n]$ in $\Theta(n)$ time using $O(n)$  words of space.
\end{lemma}

\begin{proof}
Since each instruction takes constant time, the running time is proportional
to the number of stack operations, which is $O(n)$ since each text position
is added to the stack exactly once. The space usage is dominated by the
arrays $LF$, \LA, and by the stack that use $O(n)$ words in total. \qedwhite
\end{proof}

\ignore{ It is easy to see that the maximum number of entries simultaneously
on the stack is the length of the longest decreasing subsequence $\ISA[i_1] >
\ISA[i_2] > \cdots > \ISA[i_k]$ with $i_1 < i_2 < \cdots < i_k$ and with the
additional restriction that for $j=1,\ldots,k$, if $i_1 < h < i_j$ then
$\ISA[h] > \ISA[i_j]$. If we can assume $\ISA$ is a random permutation, the
maximum stack size is $O(\sqrt{n})$ on average (see~\cite{AD99}).
Unfortunately, in the worst case, eg when $\ISA = (n,n-1,\ldots,1])$, the
maximum stack size is $n$.}

\section{Experiments}\label{s:experiments}

In this section we compare our algorithm with the
linear time algorithms of Hohlweg and
Reutenauer~\cite{Hohlweg2003,Franek2016} (\nsv) and Baier~\cite{Baier2016}
(\baier). All algorithms were adapted to compute only the Lyndon array \LA
for an input string $T[1,n]$.
In order to compare our solution with the others, we compute the suffix array
\SA for the input string $T$, then we obtain $L$ and the \LF array, and
finally we construct the Lyndon array during Burrows-Wheeler inversion
(Algorithm~\ref{a:lyndon-bwt}). This procedure will be called \our. We
used algorithm SACA-K~\cite{Nong2013} to construct \SA in $O(n)$ time using
$O(\sigma)$ working space. $\LA[1,n]$ was computed in the same space as
$\SA[1,n]$ (overwriting the values) both in \our and in \nsv.

We implemented all algorithms in ANSI C.  The source code is publicly
available at \url{https://github.com/felipelouza/lyndon-array}. The
experiments were executed on a 64-bit Debian GNU/Linux 8 (kernel 3.16.0-4)
system with an Intel Xeon Processor E5-2630 v3 20M Cache $2.40$-GHz, $384$~GB
of internal memory and a $13$ TB SATA storage.  The sources were compiled by
GNU GCC version 4.9.2, with the optimizing option -O3 for all algorithms. The
time was measured using the \texttt{clock()} function of C standard libraries
and the peak memory usage was measured using malloc\_count
library\footnote{\url{http://panthema.net/2013/malloc\_count}}.

We used string datasets from Pizza \& Chili\footnote{\url{https://pizzachili.dcc.uchile.cl/}}
as shown in the first three columns of Tables~\ref{t:exp} and~\ref{t:steps}.
The datasets
einstein-de, kernel, fib41 and cere are highly repetitive texts
The dataset english.1gb is the first 1GB of the original english dataset. In
our experiments, each integer array of length $n$ is stored using $4n$ bytes,
and each string of length $n$ is stored using $n$ bytes.

Table~\ref{t:exp} shows the running time (in seconds), the peak space memory
(in bytes per input symbol) and the working space (in GB) of each algorithm.

\paragraph{Running time}
The fastest algorithm was \baier, which overall spent about two-thirds of the time
required by \our, though the timings were much closer for larger alphabets.
\nsv was slightly faster than \our, requiring about $81\%$ of the time spent
by \our on average.

\paragraph{Peak space}
The smallest peak space was obtained by \our and \nsv, which both use
slightly more than $9n$~bytes. \our uses $9n$ bytes to store the
string $T$ and the integer arrays \SA and \LF, plus the space used by the
stack, which occupied about $11$ KB in all experiments, except for dataset
cere, in which the stack used $261$ KB. The strings $L[1,n]$ and $T[1,n]$ are
computed and decoded in the same space. \nsv also requires $9n$ bytes to
store the string $T$ and the integer arrays \SA and \ISA, that plus the space
of the stack used to compute \NSV~\cite{Goto2013}, which used exactly the
same amount of memory used by the stack of \our. The array \NSV is computed
in the same space as \ISA.
\baier uses
$17n$ bytes to store $T$, \LA and three auxiliary integer arrays of size $n$.

\paragraph{Working space}
The working space is the peak space not counting the space used by the
input string $T[1,n]$ and the output array $\LA[1,n]$ ($5n$ bytes).  The
working space of \our and \nsv were by far the smallest in all experiments.
Both algorithms use about $41\%$ of the working space used by \baier.
For dataset proteins, \our and \nsv use $7.72$ GB less memory than \baier.

\begin{table}[t]
\centering
\caption{
Experiments with Pizza \& Chili datasets. 
The datasets einstein-de, kernel, fib41 and cere are highly repetitive texts.
The running time is shown in seconds. 
The peak space is given in bytes per input symbol.
The working space is given in GB.
}
\label{t:exp}
\resizebox{1\textwidth}{!}{
\begin{tabular}{lrr||ccc||ccc||ccc}
\toprule
 &        &       & \multicolumn{3}{c||}{running time}            & \multicolumn{3}{c||}{peak space}               & \multicolumn{3}{c}{working space}              \\
 &        &       & \multicolumn{3}{c||}{[secs]}            & \multicolumn{3}{c||}{[bytes/n]}               & \multicolumn{3}{c}{[GB]}              \\
            & $\sigma$ & $n/2^{20}$     & \rotatebox{90}{\our} & \rotatebox{90}{\nsv} & \rotatebox{90}{\baier~~}   & \rotatebox{90}{\our}    & \rotatebox{90}{\nsv} & \rotatebox{90}{\baier} & \rotatebox{90}{\our}    & \rotatebox{90}{\nsv} & \rotatebox{90}{\baier} \\
\hline
sources     & 230    & 201   & \p68 & \p55 & \p\textbf{57} & \textbf{9} & \textbf{9}  & 17 & \textbf{0.79} & \textbf{0.79} & \p2.36    \\
dblp        & 97     & 282   &  104 & \p87 & \p\textbf{90} & \textbf{9} & \textbf{9}  & 17 & \textbf{1.10} & \textbf{1.10} & \p3.31    \\
dna         & 16     & 385   &  198 &  160 &  \textbf{113} & \textbf{9} & \textbf{9}  & 17 & \textbf{1.50} & \textbf{1.50} & \p4.51    \\
english.1gb & 239    & 1,047 &  614 &  504 &  \textbf{427} & \textbf{9} & \textbf{9}  & 17 & \textbf{4.09} & \textbf{4.09} & 12.27     \\
proteins    & 27     & 1,129 &  631 &  524 &  \textbf{477} & \textbf{9} & \textbf{9}  & 17 & \textbf{4.41} & \textbf{4.41} & 13.23     \\
\hdashline
einstein.de & 117    & 88    & \p36 & \p32 & \textbf{\p25} & \textbf{9} & \textbf{9}  & 17 & \textbf{0.35} & \textbf{0.35} & \p1.04    \\
kernel      & 160    & 246   &  100 & \p75 & \textbf{\p73} & \textbf{9} & \textbf{9}  & 17 & \textbf{0.96} & \textbf{0.96} & \p2.88    \\
fib41       & 2      & 256   &  120  & \p93 & \textbf{\p18}& \textbf{9} & \textbf{9}  & 17 & \textbf{1.00} & \textbf{1.00} & \p2.99    \\
cere        & 5      & 440   &  215  & 169 &  \textbf{114} & \textbf{9} & \textbf{9}  & 17 & \textbf{1.72} & \textbf{1.72} & \p5.16    \\
\bottomrule
\end{tabular}
}
\end{table}

\paragraph{Steps (running time)}
Table~\ref{t:steps} shows the running time (in seconds) for each step of
algorithms \our and \nsv.
Step 1, constructing \SA, is the most time-consuming part of both algorithms,
taking about $80\%$ of the total time. Incidentally, this means that if
the input consists of the BWT rather than $T$, our algorithm would clearly be
the fastest. In Step 2, computing \BWT is faster than computing \ISA since
$\L[i]=T[\SA[i]-1]$ is more cache-efficient than $\ISA[\SA[i]]=i$.
Similarly in Step 3, computing \LF is more
efficient than computing \NSV~\cite{Goto2013}.
However, Step 4 of \our, which computes \LA during the Burrows-Wheeler inversion, is
sufficiently slower (by a factor of $10^2$) than computing \LA from \ISA and \NSV, so that the
overall time of \our is larger than \nsv, as shown in Table~\ref{t:exp}.

\begin{table}[t]
\centering
\caption{
Experiments with Pizza \& Chili datasets. 
The running time is reported in seconds for each step of algorithms \our and \nsv.
}
\label{t:steps}
\resizebox{1\textwidth}{!}{
\begin{tabular}{lrr||c|cc|cc|cc}
\toprule
            &   &       & ~~Step 1~~ & \multicolumn{2}{c|}{Step 2} & \multicolumn{2}{c|}{Step 3} & \multicolumn{2}{c}{Step 4} \\
            &   &       &         & \rotatebox{90}{\our~~}             & \rotatebox{90}{\nsv~~}    & \rotatebox{90}{\our}             & \rotatebox{90}{\nsv}    & \rotatebox{90}{\our}     & \rotatebox{90}{\nsv}            \\
            & $\sigma$ & $n/2^{20}$ & \SA    & \BWT             & \ISA    & \LF              & \NSV    & \LA      & \LA             \\
\hline
sources     & 230    & 201   & \p50.27  & \textbf{\p2.28} & \p3.49 & \textbf{0.97} & \p1.65   & \p14.12 &   \textbf{0.13}  \\
dblp        & 97     & 282   & \p79.83  & \textbf{\p4.02} & \p5.51 & \textbf{1.46} & \p1.61   & \p18.80 &   \textbf{0.18} \\
dna         & 16     & 385   &  145.02  & \textbf{\p8.07} & \p9.99 & \textbf{1.48} & \p4.04   & \p43.39 &   \textbf{0.25} \\
english.1gb & 239    & 1,047 &  459.29  & \textbf{ 21.86} &  34.35 & \textbf{4.70} &  10.31   &  127.65 &   \textbf{0.72} \\
proteins    & 27     & 1,129 &  478.13  & \textbf{ 21.96} &  34.99 & \textbf{4.71} &  10.47   &  125.73 &   \textbf{0.75} \\
\hdashline
einstein-de & 117    & 88    & \p28.79  & \textbf{\p1.27} & \p1.91 & \textbf{0.48} & \p0.85   &  \p5.10  &  \textbf{0.06}  \\
kernel      & 160    & 246   & \p68.19  & \textbf{\p3.52} & \p4.92 & \textbf{1.29} & \p2.18   &   27.21  &  \textbf{0.18} \\
fib41       & 2      & 256   & \p85.94  & \textbf{\p5.94} & \p6.55 & \textbf{1.38} & \p0.60   &   26.48  &  \textbf{0.18} \\
cere        & 5      & 440   &  153.89  & \textbf{\p9.30} &  11.20 & \textbf{2.24} & \p4.38   &   49.38  &  \textbf{0.42} \\
\bottomrule
\end{tabular}
}
\end{table}

\section{Balanced parenthesis representation of a Lyndon Array}\label{s:bp}

In this section we introduce a new representation for the Lyndon array $\LA[1,n]$ of
$T[1,n]$ consisting of a balanced parenthesis string of length $2n$. The
existence of this representation is not surprising in view of Observation~3
in~\cite{Franek2016} stating that Lyndon words do not overlap (see also the
bracketing algorithm in~\cite{SawadaR03}). Algorithm~\ref{a:bp-construction}
gives an operational strategy for building such a representation, and the
next lemma shows how to use it to retrieve individual values of \LA.
In the following, given a balanced parenthesis string $S$, we write
$\selo(S,i)$ (resp. $\findc(S,i)$) to denote the position in $S$ of the
$i$-th open parenthesis (resp. the position in $S$ of the closed parenthesis
closing the $i$-th open parenthesis).

\ignore{Given $\LA[1,n]$ its balanced parenthesis representation
$\LABP[1,2n]$ obtained writing for $i=1,\ldots,n$ an open parenthesis
followed by a number of closed parenthesis equal to $i_c = |\{j<i |
j+\LA[j]=i+1\}|$; note $i_c$ is the number of Lyndon word ending at potion
$i$. For example, for $T=\texttt{banana\$}$ it is $\LA = [1, 2, 1, 2, 1,1,1]$
and the corresponding parenthesis representation is $\LABP =
\opa\cpa\opa\opa\cpa\cpa\opa\opa\cpa\cpa\opa\cpa\opa\cpa$.}

\begin{algorithm}[t]
\SetNlSty{textbf}{}{} \SetAlgoLined \SetCommentSty{mycommfont}

$\LABP \leftarrow \varepsilon$

$Stack \leftarrow \emptyset$

\For{$i \leftarrow 1\mbox{ to }n$}{

    \While{$Stack.top() > \ISA[i]$}{
        $Stack.pop()$

        $\LABP.append(\;\cpab\;)$  
    }
    $Stack.push(\ISA[i])$

    $\LABP.append(\;\opab\;)$   

}

$Stack.pop()$

$\LABP.append(\;\cpab\;)$

\caption{Balanced parenthesis representation $\LABP$ from $\ISA$}
\label{a:bp-construction}
\end{algorithm}

\begin{lemma}\label{lemma:bp}

The balanced parenthesis array $\LABP$ computed by
Algorithm~\ref{a:bp-construction} is such that setting for $i=1,\ldots,n$
\begin{equation}\label{eq:openclose}
o_i = \selo(\LABP,i), \qquad c_i = \findc(\LABP,i)
\end{equation}
then
\begin{equation}\label{eq:openclose2}
\LA[i] =  (c_i - o_i + 1) / 2
\end{equation}
\end{lemma}

\begin{proof}
First note that at the $i$-th iteration we append an open parenthesis to
$\LABP$ and add the value $\ISA[i]$ to the stack. The value $\ISA[i]$ is
removed from the stack as soon as a smaller element $\ISA[j] < \ISA[i]$ is
encountered. Since the last value $\ISA[n] = 1$ is the smallest element, at
the end of the for loop the stack only contains the value 1, which is removed
at the exit of the loop. Observing that we append a closed parenthesis to
$\LABP$ every time a value is removed from the stack, at the end of the
algorithm $\LABP$ indeed contains $n$ open and $n$ closed parentheses.
Because of the use of the stack, the closing parenthesis follow a first-in
last-out logic so the parenthesis are balanced.

By construction, for $i<n$, the closed parenthesis corresponding to $\ISA[i]$
is written immediately before the open parenthesis corresponding to
$\NSV_{\ISA}[i]$. Hence, between the open and closed parenthesis
corresponding to $\ISA[i]$ there is a pair of open/closed parenthesis for
each entry $k$, $i < k < \NSV_{\ISA}[i]$. Hence, using the
notation~\eqref{eq:openclose} and Lemma~\ref{l:suffixes} it is
$$
c_i - o_i - 1 \;=\; 2(\NSV_{\ISA}[i] -  \ISA[i] -1 )
            \;=\; 2 (\LA[i] - 1).
$$
which implies~\eqref{eq:openclose2}. Finally, for $i=n$ we have $o_n = 2n-1$
and $c_n=2n$, so $(c_n-o_n+1)/2=\LA[n]=1$ and the lemma follows. \qedwhite
\end{proof}

Using the range min-max tree from~\cite{NStalg12} we can represent $\LABP$ in
$2n + o(n)$ bits of space and support $\selo$, and $\findc$ in $O(1)$ time.
We have therefore established the following result.

\begin{theorem}\label{theo:bp}
It is possible to represent the Lyndon array for a text $T[1,n]$ in $2n+o(n)$
bits such that we can retrive every value $\LA[i]$ in $\Oh(1)$ time. \qedwhite
\end{theorem}

Since the new representation takes $O(n)$ bits, it is desirable to build it
without storing explicitly \ISA, which takes $\Theta(n)$ words. In
Section~\ref{s:algorithm} we used the $LF$ map to generate the \ISA\ values
right-to-left (from $\ISA[n]$ to $\ISA[1]$). Since in
Algorithm~\ref{a:bp-construction} we need to generate the \ISA\ values
left-to-right, we use the inverse permutation of the $LF$ map, known in the
literature as the $\Psi$ map. Formally, for $i=1,\ldots,n$ $\Psi[i]$ is
defined as
\begin{equation}\label{eq:psidef}
\Psi[i] =
\begin{cases}
\ISA[1]      & \mbox{if $i=1$}\\
\ISA(\SA[i]+1) & \mbox{otherwise}.
\end{cases}
\end{equation}

\begin{lemma}\label{lemma:psi}
Assume we have a data structure supporting the $\sel$ operation on the BWT in
$O(s)$ time. Then, we can generate the values $\ISA[1], \ldots, \ISA[n]$ in
$O(sn)$ time using additional $O(\sigma \log n)$ bits of space.
\end{lemma}

\begin{proof}
By~\eqref{eq:psidef} it follows that $\ISA[1] = \Psi(1)$ and, for
$i=2,\ldots,n$, $\ISA[i] = \Psi(\ISA[i-1])$. To prove the lemma we need to
show how to compute each $\Psi(i)$ in $O(s)$ time. By definition, $\Psi(i)$
is the position in \L of the character prefixing row $i$ in the conceptual
matrix defining the BWT. Let $F[1,n]$ denote the binary array such that
$F[j]=1$ iff row $j$ is the first row of the BWT matrix prefixed by some
character~$c$. Then, the character prefixing row $i$ is given by $c_i =
\rank_1(F,i)$ and
$$
\Psi(i) = \sel_{c_i}(\L,i-\sel_1(F,c_i)+1).
$$
The thesis follows observing that using~\cite{talg/RamanRS07} we can
represent $F$ in $\log{{n}\choose{\sigma}} + o(\sigma) + o(\log\log n) =
O(\sigma \log n)$ bits supporting constant time $\rank/\sel$ queries. \qedwhite
\end{proof}

\begin{lemma} \label{lemma:bpfast}
Using Algorithm~\ref{a:bp-construction} we can compute $\LABP$ from the BWT
in $O(n)$ time using $O(n)$ words of space.
\end{lemma}

\begin{proof}
We represent \L using one of the many available data structures taking
$O(n\log\sigma)$ bits and supporting constant time $\sel$ queries
(see~\cite{BCGNNalgor13} and references therein). By Lemma~\ref{lemma:psi} we
can generate the values $\ISA[1], \ldots, \ISA[n]$ in $O(n)$ overall time
using $O(\sigma\log n)$ auxiliary space. Since every other operations takes
constant time, the running time is proportional to the number of stack
operations which is $O(n)$ since each $\ISA[i]$ is inserted only once in the
stack. \qedwhite
\end{proof}

Note that the space usage of Algorithm~\ref{a:bp-construction} is dominated
by the stack, which uses $n$ words in the worst case. Since at any given time
the stack contains an increasing subsequence of $\ISA$, if we can assume that
\ISA is a random permutation the average stack size is  $O(\sqrt{n})$ words
(see~\cite{AD99}).

\ignore{It is easy to see that the maximum number of entries simultaneously
on the stack is the length of the longest decreasing subsequence $\ISA[i_1] >
\ISA[i_2] > \cdots > \ISA[i_k]$ with $i_1 < i_2 < \cdots < i_k$ and with the
additional restriction that for $j=1,\ldots,k$, if $i_1 < h < i_j$ then
$\ISA[h] > \ISA[i_j]$. If we can assume $\ISA$ is a random permutation, the
maximum stack size is $O(\sqrt{n})$ on average (see~\cite{AD99}).
Unfortunately, in the worst case, eg when $\ISA = (n,n-1,\ldots,1])$, the
maximum stack size is $n$.}

We now present an alternative representation for the stack that only uses
$n+o(n)$ bits in the worst case and supports pop and push operations in
$O(\log n/\log\log n)$ time. We represent the stack with a binary array
$S[1,n]$ such that $S[1]=1$ iff the value $i$ is currently in the stack.
Since the values in the stack are always in increasing order, $S$ is
sufficient to represent the current status of the stack. In
Algorithm~\ref{a:bp-construction} when a new element $e$ is added to the
stack we must first delete the elements larger than $e$. This can be
accomplished using rank/select operations. If $r_e = \rank_1(S,e)$ the
elements to be deleted are those returned by $\sel_1(S, r_e+i)$ for
$i=1,2,\ldots,\rank_1(S,n) - r_e$. Summing up, the binary array $S$ must
support the rank/select operations in addition to changing the value of a
single bit. To this end we use the dynamic array representation described
in~\cite{spire/HeM10} which takes $n + o(n)$ bits and support the above
operations in (optimal) $O(\log n/\log\log n)$ time. We have therefore
established, this new time/space tradeoff for Lyndon array construction.

\begin{lemma} \label{lemma:bpminute}
Using Algorithm~\ref{a:bp-construction} we can compute $\LABP$ from the BWT
in $O(n\log n/\log\log n)$ time using $O(n\log\sigma)$ bits of space. \qedwhite
\end{lemma}

Finally, we point out that if the input consists of the text $T[1,n]$ the
asymptotic costs do not change, since we can build the BWT from $T$ in $O(n)$
time and $O(n\log\sigma)$ bits of space~\cite{soda/MunroNN17}.

\begin{theorem} \label{th:bp}
Given $T[1,n]$ we can compute $\LABP$ in $O(n)$ time using $O(n)$ words of
space, or in $O(n\log n/\log\log n)$ time using $O(n\log\sigma)$ bits of
space. \qedwhite
\end{theorem}

\section{Summary of Results}\label{s:conclusion}

In this paper we have described a previously unknown connection between the
Burrows-Wheeler transform and the Lyndon array, and proposed a corresponding
algorithm to construct the latter during Burrows-Wheeler inversion.
The algorithm is guaranteed linear-time and simple, resulting in the good
practical performance shown by the experiments.

Although not faster than other linear algorithms, our solution was one
of the most space-efficient. In addition, if the input is stored in a
BWT-based self index, our algorithm would have a clear advantage in both
working space and running time, since it is the only one that uses the LF-map
rather than the suffix array.

We also introduced a new balanced parenthesis representation for the
Lyndon array using $2n+o(n)$ bits supporting $O(1)$ time access. We have
shown how to build this representation in linear time using $O(n)$ words of
space, and in $O(n\log n/\log\log n)$ time using asymptotically the same
space as $T$.

Over all the known algorithms surveyed in~\cite{Smyth2017}, probably the
fastest for real world datasets and the most space-efficient is the
folklore \maxlyn algorithm described in \cite{Franek2016}, which makes no use
of suffix arrays and requires only constant additional space, but which
however requires $\Theta(n^2)$ time in the worst-case. We tested \maxlyn
on a string consisting of $10 \times 2^{20}$ symbols `a'. While the
linear-time algorithms run in no more than $0.5$ seconds, \maxlyn takes about
$8$ hours to compute the Lyndon array. Thus, the challenge that remains is
to find a fast and ``lightweight'' worst-case linear-time algorithm for
computing the Lyndon array that avoids the expense of suffix array
construction.

\section*{Acknowledgments}
The authors thank Uwe Baier for kindly providing the source code of
algorithm \baier, and Prof. Nalvo Almeida for granting access to the machine used
for the experiments.

\paragraph{Funding}
F.A.L. was supported by the grant $\#$2017/09105-0 from the São Paulo Research Foundation (FAPESP).
The work of W.F.S. was supported in part by a grant from the Natural Sciences
\& Engineering Research Council of Canada (NSERC). G.M. was supported by the
PRIN grant 201534HNXC and INdAM-GNCS Project {\sl Efficient algorithms for
the analysis of Big Data}. G.P.T. acknowledges the support of CNPq.


\end{document}